\documentclass[conference]{IEEEtran}
\ifCLASSINFOpdf
\else
\fi

\usepackage{amsfonts}
\usepackage{graphicx,subfig}
\usepackage{amsmath,amsthm}
\usepackage{wrapfig}

\usepackage{epsfig}
\usepackage{amsmath}
\usepackage{amssymb}
\usepackage{psfrag}
\usepackage[absolute]{textpos}

\newtheorem{thm}{Theorem}
\newtheorem{lem}{Lemma}
\newtheorem{definition}{Definition}

\newtheorem{algorithm}{Algorithm}

\newcommand{\A}{{\bf A}}
\def\x{{\bf x}}
\def\y{{\bf y}}
\def\w{{\bf w}}

\def\z{{\bf z}}

\newcommand{\beq}{\begin{equation}}
\newcommand{\eeq}{\end{equation}}
\newcommand{\bea}{\begin{eqnarray}}
\newcommand{\eea}{\end{eqnarray}}

\newcommand{\Prob}{\ensuremath{\mathbb{P}}}

\long\def\symbolfootnote[#1]#2{\begingroup%
\def\thefootnote{\fnsymbol{footnote}}\footnote[#1]{#2}\endgroup}

\begin{document}
%
\title{On the Scaling Law for Compressive Sensing and its Applications}


\author{\authorblockN{Weiyu Xu}
\authorblockA{Cornell University \\
 Ithaca NY 14853, USA \\
 Email: wx42@cornell.edu}%
\and
\authorblockN{Ao Tang}
\authorblockA{Cornell University\\
Ithaca NY 14853, USA \\
Email: atang@ece.cornell.edu}

}


%


\maketitle

\begin{abstract}
\boldmath
$\ell_1$ minimization can be used to recover sufficiently sparse unknown signals from compressed linear measurements. In fact, exact thresholds on the sparsity (the size of the support set), under which with high probability  a sparse signal can be
recovered from i.i.d. Gaussian measurements, have been computed and are
referred to as ``weak thresholds'' \cite{D}. It was also known that there is a tradeoff between the sparsity and the $\ell_1$ minimization recovery stability. In this paper, we give a \emph{closed-form} characterization for this tradeoff which we call the scaling law for
compressive sensing recovery stability. In a nutshell, we are able to show that as the sparsity backs off $\varpi$ ($0<\varpi<1$) from the weak threshold of $\ell_1$ recovery, the parameter for the recovery stability will scale as $\frac{1}{\sqrt{1-\varpi}}$.
Our result is based on a careful analysis through the Grassmann angle framework for the Gaussian measurement matrix. We will further discuss how this scaling law helps in analyzing the iterative reweighted $\ell_1$ minimization algorithms. If the nonzero elements over the signal support follow an amplitude probability density function (pdf) $f(\cdot)$ whose $t$-th derivative $f^{t}(0) \neq 0$ for some integer $t \geq 0$, then a certain iterative reweighted $\ell_1$ minimization algorithm can be analytically shown to lift the phase transition thresholds (weak thresholds) of the plain $\ell_1$ minimization algorithm.
\end{abstract}
\IEEEpeerreviewmaketitle

\section{Introduction}
\label{sec:Intro}
Compressive sensing addresses the problem of recovering sparse signals
from under-determined systems of linear equations  \cite{Rice}. In
particular, if $\x$ is an $n\times1$ real-numbered vector that is known to have at
most $k$ nonzero elements where $k<n$, and $\A$ is an $m\times n $
measurement matrix with $k<m<n$, then for appropriate values of $k$,
$m$ and $n$, it is possible to efficiently recover $\x$ from $\y=\A\x$ \cite{D
  CS,DT,CT,Baraniuk Manifolds}. The most well recognized powerful recovery
algorithm is $\ell_1$ minimization which can be formulated  as
follows:
\beq
\label{eq:l1 min}
\min_{\A\z=\A\x}\|\z\|_1
\eeq
The first result that established the fundamental phase transitions of signal
recovery using $\ell_1$ minimization is due to Donoho and Tanner
\cite{DT,D}, where it was shown
that if the measurement matrix is i.i.d. Gaussian, for a given ratio
of $\delta = \frac{m}{n}$, $\ell_1$ minimization can successfully
recover {\em every} $k$-sparse signal, provided that $\mu = \frac{k}{n}$ is
smaller that a certain threshold. This statement is true
asymptotically as $n \rightarrow \infty$ and with high
probability. This threshold guarantees the recovery of {\em all}
sufficiently sparse signals and is therefore referred to as a ``strong''
threshold. It therefore does not depend on the actual distribution of
the nonzero entries of the sparse signal and thus is a universal
result. 

Another notion introduced and computed in \cite{DT,D} is that of a
{\em weak} threshold $\mu_{W}(\delta)$ under which signal recovery is guaranteed for {\em
  almost all} support sets and {\em almost all} sign patterns of the
sparse signal, with high probability as $n\rightarrow\infty$. The weak
threshold is the one that can be observed in simulations of $\ell_1$
minimization and allows for signal recovery beyond the strong threshold. It
is also universal in the sense that it applies to any amplitude that the nonzero signal entries take.

When the sparsity of the signal $\x$ is larger than the weak threshold $\mu_{W}(\delta)n$, a common stability result for the $\ell_1$ minimization is that, for a set $K \subseteq \{1,2, ..., n\}$ with cardinality $|K|$ small enough for $A$ to satisfy the restrict isometry condition \cite{CT} or the null space robustness property \cite{devore2} \cite{XuHassibiAllerton08},  the decoding error is bounded by,
\begin{equation}
\|\x-\hat{\x}\|_1 \leq D \|\x_{\overline{K}}\|_1,
\label{eq:generalstability}
\end{equation}
where $\hat{\x}$ is any minimizer to $\ell_1$ minimization, $D$ is a constant, $\overline{K}$ is the complement of the set $K$ and $\x_{\overline{K}}$ is the part of $\x$ over the set $\overline{K}$.

To date, known bounds on $|K|/n$, for the restricted isometry condition to hold with overwhelming probability, are small compared with the weak threshold $\mu_{W}(\delta)$ \cite{CT}. \cite{isitrobust} \cite{XuHassibiAllerton08} used the Grassmann angle approach to characterize sharp bounds on the stability of $\ell_1$ minimization and showed that, for an arbitrarily small $\epsilon_{0}$, as long as $|K|/n=(1-\epsilon_{0})\mu_{W}(\delta)n$, with overwhelming probability as $n \rightarrow \infty$, (\ref{eq:generalstability}) holds for some constant $D$ ($D$ of course depends on $|K|/n$). However, no closed-form formula for $D$ were given.

In this paper, we give a \emph{closed-form} characterization for this tradeoff which we call the scaling law for compressive sensing recovery stability. Namely, we will give a closed-form bound for $D$ as a function of $|K|/n$. It is the first result of such kind.
This result is obtained from close analysis through the Grassmann angle framework for the Gaussian measurement matrix. We will further discuss how this scaling law helps in analyzing the iterative reweighted $\ell_1$ minimization algorithm.

Using this scaling law results for the stability and the Grassmann angle framework for the weighted $\ell_1$ minimization, we prove that a certain \emph{iterative reweighted $\ell_1$} algorithm indeed has better weak recovery guarantees for particular classes of sparse signals, including sparse Gaussian signals. We previously introduced these algorithms in \cite{Icassp reweighted l_1}, and had proven that for a very restricted class of sparse signals they outperform standard $\ell_1$ minimization. In this paper, we are able to extend this result to a much wider and more reasonable class of sparse signals. The key to our result is the fact that for these signals, $\ell_1$ minimization has an \emph{approximate support recovery} property \cite{Isit reweighted l_1} which can be exploited by a reweighted $\ell_1$ algorithm, to obtain a provably superior weak threshold.
More specifically, if the nonzero elements over the signal support follow a probability density function (pdf) $f(\cdot)$ whose $t$-th derivative $f^{t}(0) \neq 0$ for some $t \geq 0$, then a certain iterative reweighted $\ell_1$ minimization algorithm can be analytically shown to lift the phase transition thresholds (weak thresholds) of the plain $\ell_1$ minimization algorithm through using the scaling law for the sparse recovery stability.
This extends our earlier results of weak threshold improvements for sparse vectors with nonzero elements following the Gaussian distribution, whose pdf is itself nonzero at the origin (namely its $0$-th derivative is nonzero) \cite{Isit reweighted l_1}.

It is worth noting that different variations of reweighted $\ell_1$
algorithms have been recently introduced in the literature and, have
shown experimental improvement over ordinary $\ell_1$ minimization
\cite{CWB07,Needell}.  In \cite{Needell} approximately sparse signals
have been considered, where perfect recovery is never
possible. However, it has been shown that the recovery noise can be
reduced using an iterative scheme. In \cite{CWB07}, a similar
algorithm is suggested and is empirically shown to outperform $\ell_1$
minimization for exactly sparse signals with non-flat
distributions. Unfortunately, \cite{CWB07} provides no theoretical
performance guarantee. 

This paper is organized as follows. In Section \ref{sec:basic} and \ref{sec:model}, we introduce the basic concepts and system model. In Section \ref{sec:scaling}, we introduce and derive the main result of this paper: the scaling law for the compressive sensing recovery stability. In the following sections, we will use the scaling law to give new analysis results about the iterative reweighted $\ell_1$ minimization algorithms. 

\hfill 
\section{Basic Definitions}
\label{sec:basic}
 A sparse signal with exactly $k$ nonzero entries is called
 $k$-sparse. For a vector $\x$, $\|\x\|_1$ denotes the $\ell_1$
 norm. The support (set) of $\x$,  denoted by $supp(\x)$, is the index
 set of its nonzero coordinates. For a vector $\x$ that is not exactly
 $k$-sparse, we define the $k$-support of $\x$ to be the index set of
 the largest $k$ entries of $\x$ in amplitude, and denote it by
 $supp_k(\x)$. For a subset $K$ of the entries of $\x$, $\x_K$ means
 the vector formed by those entries of $\x$ indexed in $K$. Finally,
 $\max|\x|$ and $\min|\x|$ mean the absolute value of the maximum and
 minimum entry of $\x$ in magnitude, respectively.

\section{Signal Model and Problem Description}
\label{sec:model}
We consider sparse random signals with i.i.d. nonzero
entries. In other words we assume that the unknown sparse signal is an
$n\times 1$ vector $\x$ with exactly $k$ nonzero entries, where each
nonzero entry is independently sampled from a well defined distribution. The measurement matrix $\A$ is a $m\times n$
matrix with i.i.d. Gaussian entries with a ratio of dimensions $\delta
= \frac{m}{n}$. Compressed sensing theory guarantees that if
$\mu=\frac{k}{n}$ is smaller than a certain threshold, then every
$k$-sparse signal can be recovered using $\ell_1$ minimization. The
relationship between $\delta$ and the maximum threshold of $\mu$ for
which such a guarantee exists is called the \emph{strong sparsity
  threshold}, and is denoted by $\mu_{S}(\delta)$. A more practical
performance guarantee is the so-called \emph{weak sparsity threshold},
denoted by $\mu_{W}(\delta)$, and has the following
interpretation. For a fixed value of $\delta = \frac{m}{n}$ and
i.i.d. Gaussian matrix $\A$ of size $m\times n$,  a random $k$-sparse
vector $\x$ of size $n\times 1$ with a randomly chosen support set and
a random sign pattern can be recovered from $\A\x$ using $\ell_1$
minimization with high probability, if $\frac{k}{n}<\mu_{W}(\delta)$. Similar recovery thresholds can be
obtained by imposing more or less restrictions. For example, strong
and weak thresholds for nonnegative signals have been evaluated in
\cite{Donoho positive}.

We assume that the support size of $\x$, namely $k$, is slightly
larger than the weak threshold of $\ell_1$ minimization. In other
words,  $k = (1+\epsilon_0)\mu_{W}(\delta)$ for some
$\epsilon_0>0$. This means that if we use $\ell_1$ minimization, a
randomly chosen $\mu_{W}(\delta)n$-sparse signal will be recovered
perfectly with very high probability, whereas a randomly selected
$k$-sparse signal will not. We would like to show that for a strictly
positive $\epsilon_0$, the iterative reweighted $\ell_1$ algorithm of
Section \ref{sec:Algorithm} can indeed recover a randomly selected
$k$-sparse signal with high probability, which means that it has an
improved weak threshold.

\section{The Scaling Law for the Compressive Sensing Stability}
\label{sec:scaling}
In this section, we will derive the scaling of the $\ell_1$ recovery stability as a function of the signal sparsity. More specifically, we are interested in characterizing a closed-form relationship between $C$ and the sparsity $|K|$ in the following theorem.

\begin{thm}
Let $A$ be a general $m\times n$ measurement matrix, $\x$ be an $n$-element
vector and $\y=A\x$. Denote $K$ as a subset of $\{1,2,\dots,n\}$ such
that its cardinality $|K|=k$ and further denote $\overline{K}=\{1,2,\dots,n\}\setminus K$. Let $\w$
denote an $n \times 1$ vector. Let $C>1$ be a fixed number.

Given a specific set $K$ and suppose that the part of $\x$ on $K$, namely $\x_{K}$ is fixed.
$\forall \x_{\overline{K}}$, any solution $\hat{\x}$ produced by the $\ell_1$ minimization
satisfies
\begin{equation*}
 \|\x_K\|_1-\|\hat{\x}_K\|_{1} \leq
\frac{2}{C-1} \|\x_{\overline{K}}\|_1
\end{equation*}
 and
\begin{equation*}
 \|(\x-\hat{\x})_{\overline{K}}\|_1
\leq \frac{2C}{C-1} \|\x_{\overline{K}}\|_1,
\end{equation*}
if and only if $\forall \w\in \mathbb{R}^n~\mbox{such that}~A\w=0$,
we have
 \begin{equation}
\|\x_K+\w_{K}\|_1+ \|\frac{\w_{\overline{K}}}{C}\|_1 \geq \|\x_K\|_1.
 \label{eq:Grasswthmeq1}
 \end{equation}
\end{thm}

In fact, if (\ref{eq:Grasswthmeq1}) is satisfied, we will have the stability result

\begin{equation*}
 \|(\x-\hat{\x})_{\overline{K}}\|_1
\leq \frac{2C}{C-1} \|\x_{\overline{K}}\|_1.
\end{equation*}

In \cite{isitrobust}, it was established that when the matrix $A$ is sampled from an i.i.d. Gaussian ensemble, $C=1$, considering a single index set $K$, there exists a constant ratio $0<\mu_{W}<1$ such that if $\frac{|K|}{n} \leq \mu_{W}$, then with overwhelming probability as $n \rightarrow \infty$, the condition (\ref{eq:Grasswthmeq1}) holds for all $\w\in \mathbb{R}^n~\mbox{satisfying}~A\w=0$. Now if we take a single index set $K$ with cardinality $\frac{|K|}{n}=(1-\varpi){\mu_{W}}$, we would like to derive a characterization of $C$, as a function of $\frac{|K|}{n}=(1-\varpi){\mu_{W}}$, such that the condition (\ref{eq:Grasswthmeq1}) holds for all $\w\in \mathbb{R}^n~\mbox{satisfying}~A\w=0$. The main result of this paper is stated in the following theorem.

\begin{thm}
Assume the $m\times n$ measurement matrix $A  $ is sampled from an i.i.d. Gaussian ensemble, let $K$ be a single index set with $\frac{|K|}{n}=(1-\varpi){\mu_{W}}$, where ${\mu_{W}}$ is the weak threshold for ideally sparse signals and $\varpi$ is any real number between $0$ and $1$. We also let $\x$ be an $n$-dimensional signal vector with $\x_{K}$ being an arbitrary but fixed signal component. Then with overwhelming probability, the condition (\ref{eq:Grasswthmeq1}) holds for all $\w \in \mathbb{R}^n~\mbox{satisfying}~A\w=0$, with respect to the parameter $C=\frac{1}{\sqrt{1-\varpi}}$.

\end{thm}

\begin{proof}
When the measurement matrix $A$ is sampled from an i.i.d. Gaussian ensemble, it is known that the probability that the condition (\ref{eq:Grasswthmeq1}) holds for all $\w \in \mathbb{R}^n~\mbox{satisfying}~A\w=0$ is the \emph{Grassmann angle}, namely the probability that an $(n-m)$-dimensional uniformly distributed subspace intersects a polyhedral cone trivially (intersecting only at the apex of the cone). The complementary probability that the condition (\ref{eq:Grasswthmeq1}) does not hold for all $\w \in \mathbb{R}^n~\mbox{satisfying}~A\w=0$ is the \emph{complementary Grassmann angle}. In our problem, without loss of generality, we scale $\x_{K}$ (extended to an $n$-dimensional vector supported on $K$) to a point in the relative interior of a $(k-1)$-dimensional face $F$ of the weighted $\ell_1$ ball,
\begin{equation}
\text{SP}=\{\y\in \mathbb{R}^n~|~\|\y_K\|_1+ \|\frac{\y_{\overline{K}}}{C}\|_1 \leq 1\}.
\end{equation}
The polyhedral cone we are interested in for the complementary Grassmann angle is the cone $\text{SP}-\x_{K}$,  namely the cone obtained by setting $\x_{K}$  as the apex, and observing $\text{SP}$ from this apex.

Building on the works by Santal\"{o} \cite{Santalo1952} and
McMullen \cite{McMullen1975}  in high dimensional integral
geometry and convex polytopes, the complementary Grassmann angle for
the $(k-1)$-dimensional face $F$ can be explicitly expressed as the
sum of products of internal angles and external angles \cite{Grunbaumbook}:
\begin{equation}
P=2\times \sum_{s \geq 0}\sum_{G \in \Im_{m+1+2s}(\text{SP})}
{\beta(F,G)\gamma(G,\text{SP})}, \label{eq:Grassangformula}
\end{equation}
where $s$ is any nonnegative integer, $G$ is any
$(m+1+2s)$-dimensional face of the SP
($\Im_{m+1+2s}(\text{SP})$ is the set of all such faces),
$\beta(\cdot,\cdot)$ stands for the internal angle and
$\gamma(\cdot,\cdot)$ stands for the external angle.

The internal angles and external angles are basically defined as
follows \cite{Grunbaumbook}\cite{McMullen1975}:
\begin{itemize}
\item An internal angle $\beta(F_1, F_2)$ is the fraction of the
hypersphere $S$ covered by the cone obtained by observing the face
$F_2$ from the face $F_1$. \footnote{Note the dimension of the
hypersphere $S$ here matches the dimension of the corresponding cone
discussed. Also, the center of the hypersphere is the apex of the
corresponding cone. All these defaults also apply to the definition
of the external angles. } The internal angle $\beta(F_1, F_2)$ is
defined to be zero when $F_1 \nsubseteq F_2$ and is defined to be
one if $F_1=F_2$.
\item An external angle $\gamma(F_3, F_4)$ is the fraction of the
hypersphere $S$ covered by the cone of outward normals to the
hyperplanes supporting the face $F_4$ at the face $F_3$. The
external angle $\gamma(F_3, F_4)$ is defined to be zero when $F_3
\nsubseteq F_4$ and is defined to be one if $F_3=F_4$.

\end{itemize}

When $C=1$, we denote the probability $P$ in (\ref{eq:Grassangformula}) as $P_{1}$. By definition, the weak threshold $\mu_{W}$ is the supremum of $\frac{|K|}{n} \leq \mu_{W}$ such that the probability $P_{1}$ in (\ref{eq:Grassangformula}) goes to $0$ as $n \rightarrow \infty$. We need to show for $\frac{|K|}{n}=(1-\varpi){\mu_{W}}$ and $C=\frac{1}{\sqrt{1-\varpi}}$, (\ref{eq:Grassangformula}) also goes to $0$ as $n \rightarrow \infty$. To that end, we only need to show the probability $P'$ that, there exists an $\w$ from the null space of $A$ such that
 \begin{equation}
\|\x_K+\w_{K}\|_1+ \|\frac{\w_{\overline{K_1}}}{C_{\infty}}\|_1 +\|\frac{\w_{\overline{K_2}}}{C}\|_1 < \|\x_K\|_1
 \label{eq:Grasswthmeq2}
 \end{equation}
 goes to $0$ as $n \rightarrow \infty$, where $C_{\infty}$ is a large number which we may take as $\infty$ at the end, ${\overline{K_1}}$, ${\overline{K_2}}$ and $K$ are disjoint sets such that $|{\overline{K_1}}\bigcup {K}|=\mu_{W}n$ and ${\overline{K_1}}\bigcup {\overline{K_2}}=\overline{K}$.

Then the probability $P'$ will be equal to the probability that an $(n-m)$-dimensional uniformly distributed subspace intersects the polyhedral cone $\text{WSP}-\x_{K}$ nontrivially (intersecting at some other points besides the apex of the cone), where $\text{WSP}$ is the polytope

\begin{equation}
\text{WSP}=\{\y\in \mathbb{R}^n~|~\|\y_K\|_1+ \|\frac{\y_{\overline{K_1}}}{C_{\infty}}\|_1 +\|\frac{\y_{\overline{K_2}}}{C}\|_1 \leq 1\}.
\end{equation}

Then $P'$ is also a complementary Grassmann angle, which can be expressed by
\cite{Grunbaumbook}:
\begin{equation}
P'=2\times \sum_{s \geq 0}\sum_{G \in \Im_{m+1+2s}(\text{WSP})}
{\beta(F,G)\gamma(G,\text{WSP})}. \label{eq:Grassangformulanewp}
\end{equation}

Now we only need to show $P' \leq P_{1}$. If we denote $l=(m+1+2s)+1$ and $k=(1-\varpi)\mu_{W}n$, in the polytope $\text{WSP}$, then there are in total $\binom{n-k}{l-k} 2^{l-k}$ faces $G$ of dimension
$(l-1)$ such that $F\subseteq G$ and $\beta(F, G) \neq 0$.

However, we argue that  when $C_{\infty}$ is very large, only $\binom{n-k_1}{l-k_1} 2^{l-k}$ such faces $G$ of dimension $(l-1)$ will contribute nonzero terms to $P'$ in (\ref{eq:Grassangformulanewp}), where $k_1=\mu_{W}n$. In fact, a certain $(l-1)$-dimensional face $G$ supported on the index set $L$ is the convex hull of $C_{i} e_{i}$, where $i \in L $, $C_{i}$ is the corresponding weighting for index $i$ (which is $1$ for the set $K$, $C_{\infty}$ for the set $\overline{K_1}$ and $C$ for the set $\overline{K_2}$ ), and $e_i$ is the standard unit coordinate vector. Now we show that if $\overline{K_1} \nsubseteq L$, the corresponding term in (\ref{eq:Grassangformulanewp}) for the face $G$ will be $0$ when $C_{\infty}$ is very large.

\begin{lem}
Suppose that $F$ is a $(k-1)$-dimensional face of \text{WSP} supported on the subset $K$ with $|K|=k$. Then the external angle
$\gamma(G, \text{WSP})$ between an $(l-1)$-dimensional face $G$ supported on the set $L$($F
\subseteq G$) and the  polytope $\text{WSP}$ is $0$ when  $\overline{K_1} \nsubseteq L$ and $C_{\infty}$ is large.
\label{lemma:external}
\end{lem}

\begin{proof}
Without loss of generality, assume $K=\{n-k+1, \cdots,n\}$. Consider
the $(l-1)$-dimensional face
\begin{equation*}
G=\text{conv}\{C_{n-l+1}\times e^{n-l+1}, ... ,C_{n-k}\times e^{n-k}, e^{n-k+1},
..., e^{n}\}
\end{equation*}
of $\text{WSP}$. The $2^{n-l}$ outward
normal vectors of the supporting hyperplanes of the facets
containing $G$ are given by
\begin{equation*}
\{\sum_{p=1}^{n-l} j_{p}e_p/{C_p}+\sum_{p=n-l+1}^{n-k} e_p/C_p+
\sum_{p=n-k+1}^{n} e_p, j_{p}\in\{-1,1\}\}.
\end{equation*}

Then the outward normal cone $c(G, \text{WSP})$ at the face $G$ is
the positive hull of these normal vectors. When $\overline{K_1} \nsubseteq L$, the fraction of the surface of the $(n-l-1)$-dimensional sphere taken by the cone $c(G, \text{WSP})$ is $0$ since the corresponding $C_{p}$ is very large.
\end{proof}

Now let us look at the internal angle $\beta(F,G)$ between the
$(k-1)$-dimensional face $F$ and an $(l-1)$-dimensional face $G$, where $\overline{K_1}$ is a subset of the support set of $G$.
Notice that the only interesting case is when $F \subseteq G$ since
$\beta(F,G)\neq 0$ only if $F \subseteq G$. We will see if $F
\subseteq G$, the cone $c(F,G)$ formed by observing $G$ from $F$ is the
direct sum of a $(k-1)$-dimensional linear subspace and the positive hull of $(l-k)$ vectors. These $(l-k)$ vectors are in the form
\begin{equation*}
v_{i}=(-\frac{1}{k},...,-\frac{1}{k},0,...,C_{i},0,...0), i \in L \setminus K.
\end{equation*}
 For those vectors $v_{i}$ with $i \in \overline{K_1}$, $C_{i}=C_{\infty}$. When $C_{\infty}$ is very large, the considered cone takes half of the space at each $i$-th coordinate with $i \in \overline{K_1}$.

So by the definition of the internal angle, the internal angle $\beta(F,G)$ is equal to $\frac{1}{2^{k_1-k}} \times \beta{(F,G_1)}$, where $G_1$ is supported only on the set $L\setminus \overline{K_1}$. It is known that this internal angle $\beta{(F,G_1)}$ is equal to the fraction of an $(l-k_1-1)$-dimensional sphere taken by a polyhedral cone formed by $(l-k_{1})$ unit vectors with inner product
$\frac{1}{1+C^2k}$ between each other. In this case, the internal
angle is given by
\begin{equation}
\beta(F,G)=\frac{1}{2^{k_1-k}} \frac{V_{l-k_1-1}(\frac{1}{1+C^2k},l-k_1-1)}{V_{l-k_1-1}(S^{l-k_1-1})},
\label{eq:internal}
\end{equation}
where $V_i(S^i)$ denotes the $i$-th dimensional surface measure on
the unit sphere $S^{i}$, while $V_{i}(\alpha', i)$ denotes the
surface measure for regular spherical simplex with $(i+1)$ vertices
on the unit sphere $S^{i}$ and with inner product as $\alpha'$
between these $(i+1)$ vertices. Thus (\ref{eq:internal}) is equal to
$B(\frac{1}{1+C^2k}, l-k_1)$, where
\begin{equation}
B(\alpha', m')=\theta^{\frac{m'-1}{2}} \sqrt{(m'-1)\alpha' +1}
\pi^{-m'/2} {\alpha'}^{-1/2}J(m',\theta),
\end{equation}
with $\theta=(1-\alpha')/\alpha'$ and
\begin{equation}
 J(m', \theta)=\frac{1}{\sqrt{\pi}} \int_{-\infty}^{\infty}(\int_{0}^{\infty} e^{-\theta v^2+2i v\lambda} \,dv )^{m'} e^{-\lambda^2} \,d\lambda.
\end{equation}

If we take $C=\frac{1}{\sqrt{1-\varpi}}$, then
\begin{equation*}
\frac{1}{1+C^2k}=\frac{1}{1+k_1}.
\end{equation*}

By comparison, $\beta(F,G)=\frac{1}{2^{k_1-k}} \times \beta{(F,G)}$ is exactly the $\frac{1}{2^{k_1-k}} \beta(F_1, G_1)$ term appearing in the expression for the Grassmann angle $P$ between the face $F_1$ supported on the set $K_1$ and the polytope $\text{SP}$, where $G_1$ is an $(l-1)$-dimensional face of $\text{SP}$ supported on the set $L$.

Similar to the derivation for the internal angle, we can show that the external angle $\gamma(G, \text{WSP})$ is also exactly equal to $\gamma(G_1, \text{SP})$ term appearing in the expression for the Grassmann angle $P$ between the face $F_1$ supported on the set $K_1$ and the polytope $\text{SP}$, where $G_1$ an $(l-1)$-dimensional face of $\text{SP}$ supported on the set $L$.

Since there are in total only $\binom{n-k_1}{l-k_1} 2^{l-k}$ such faces $G$ of dimension $(l-1)$ will contribute nonzero terms to $P'$ in (\ref{eq:Grassangformulanewp}), substituting the results for the internal and external angles, we have $P=P'$. Thus for $\frac{|K|}{n}=(1-\varpi) \mu_{W}$ and $C =\frac{1}{\sqrt{1-\varpi}}$, with high probability, the condition the condition (\ref{eq:Grasswthmeq1}) holds for all $\w \in \mathbb{R}^n~\mbox{satisfying}~A\w=0$.

\end{proof}
\section{Iterative Weighted $\ell_1$ Algorithm}
\label{sec:Algorithm}
Beginning from this section, we will see how the stability result is used in analyzing the iterative reweighted $\ell_1$ minimization algorithms. We focus on the following algorithm from \cite{Icassp reweighted l_1, Isit reweighted l_1}, consisting of two $\ell_1$ minimization steps: a standard one and a weighted one. The input to the algorithm is the vector $\y=\A\x$, where $\x$ is a $k$-sparse
signal with $k=(1+\epsilon_0)\mu_W(\delta)n$, and the output is an
approximation $\x^*$ to the unknown vector $\x$. We assume that $k$,
or an upper bound on it, is known. Also $\omega>1$ is a
predetermined weight.
%
\begin{algorithm}\cite{Isit reweighted l_1}
\text{}
\begin{enumerate}
\item Solve the $\ell_1$ minimization problem:
\begin{equation}
\hat{\x} = \arg{ \min{ \|\z\|_1}}~~\text{subject to}~~ \A\z
= \A\x.
\end{equation}
\item Obtain an approximation for the support set of $\x$:
find the index set $L \subset \{1,2, ..., n\}$ which corresponds to
the largest $k$ elements of
$\hat{\x}$ in magnitude.
\item Solve the following weighted $\ell_1$ minimization problem and declare the solution as output:
 \beq
\x^* = \arg{\min\|\z_L\|_1+\omega\|\z_{\overline{L}}\|_1}~~\text{subject to}~~ \A\z
= \A\x.
\label{eq:weighted l_1}
\eeq
\end{enumerate}
\label{alg:modmain}
\end{algorithm}

The idea behind the algorithm is as follows. In the first step
we perform a standard $\ell_1$ minimization. If the sparsity of the
signal is beyond the weak threshold $\mu_W(\delta)n$, then $\ell_1$ minimization is not capable of recovering the signal. However, we can use its output to identify an index set $L$
in which most elements correspond to the nonzero elements of $\x$. We
finally perform a weighted $\ell_1$ minimization by penalizing those
entries of $\x$ that are not in $L$ because they have a
lower chance of being nonzero elements.

In the next sections we formally prove that, for certain classes of signals, Algorithm
\ref{alg:modmain} has a recovery threshold beyond that of standard
$\ell_1$ minimization. The idea of the proof is as follows. In Section
\ref{sec:robustness}, we prove that there is a large overlap between
the index set $L$, found in Step 2 of the algorithm, and the support set
of the unknown signal $\x$ (denoted by $K$)---see Theorem \ref{thm:l_1
  support recovery}. Then in Section
\ref{sec:perfect recovery}, we show that the large overlap between $K$
and $L$ can result in perfect recovery of $\x$, beyond the standard
weak threshold, when a weighted $\ell_1$ minimization is used in Step
3.

This proof idea was already used in \cite{Isit reweighted l_1} to prove a threshold improvement in recovering sparse vectors with Gaussian distributed nonzero elements by using a numerical evaluation of the robustness

\section{Approximate Support Recovery, Steps 1 and 2 of the Algorithm}
\label{sec:robustness}

In this section, we carefully study the first two steps of Algorithm
\ref{alg:modmain}. The unknown signal $\x$ is assumed to be a $k$-sparse vector with support set $K$, where
$k=|K|=(1+\epsilon_0)\mu_{W}(\delta)n$, for some $\epsilon_0>0$. 
The set $L$, as defined in the algorithm, is in fact the
$k$-support set of $\hat{\x}$. We show that for small enough
$\epsilon_0$, the intersection of  $L$ and $K$ is very large with high
probability, so that $L$ can be counted as a good
approximation to $K$.

We now lower bound $|L\cap K|$. First, we
state a general lemma that bounds $|K\cap L|$ as a function of
$\|\x-\hat{\x}\|_1$ \cite{Isit reweighted l_1}. Then, we recall an intrinsic property of
$\ell_1$ minimization called \emph{weak robustness} that provides an
upper bound on the quantity $\|\x-\hat{\x}\|_1$. 
\begin{definition}\cite{Isit reweighted l_1}
For a $k$-sparse signal $\x$, we define $W(\x,\lambda)$ to be the size of the largest subset of nonzero entries of $\x$ that has a $\ell_1$ norm less than or equal to $\lambda$.
\beq
\nonumber W(\x,\lambda):= \max\{|S|~|~S\subseteq supp(\x),~\|\x_S\|_1\leq \lambda\}
\eeq
\end{definition}
\noindent Note that $W(\x,\lambda)$ is increasing in $\lambda$.
\begin{lem} \cite{Isit reweighted l_1}
Let $\x$ be a $k$-sparse vector and $\hat{\x}$ be another vector. Also, let $K$ be the support set of $\x$ and $L$ be the $k$-support set of $\hat{\x}$. Then
\beq
|K\cap L|\geq k-W(\x,\|\x-\hat{\x}\|_1)
\eeq
\label{lem:deviation thm}
\end{lem}
We now review the notion of weak robustness, which
allows us to bound $\|\x-\hat{\x}\|_1$, and has the following formal
definition \cite{isitrobust}.
\begin{definition}
Let the set $S\subset\{1,2,\cdots,n\}$ and the subvector $\x_S$  be
fixed. A solution $\hat{\x}$ is called weakly robust if, for some $C >
1$ called the robustness factor, and all $\x_{\overline{S}}$, it
holds that
\beq
\|(\x-\hat{\x})_{\overline{S}}\|_1 \leq \frac{2C}{C-1}\|\x_{\overline{S}}\|_1.
\label{def1}
\eeq%
\end{definition}
\noindent The weak robustness notion allows us to bound the error
in $\|\x -\hat{\x}\|_1$ in the following way. If the matrix $\A_S$ ,
obtained by retaining only those columns of $\A$ that are
indexed by $S$, has full column rank, then the quantity
\beq
\nonumber
~\kappa = \max_{\A\w =0, \w\neq 0} \frac{\|\w_S\|_1}{\|\w_{\overline{S}}\|_1}
\eeq
must be finite, and one can write
\beq
\|\x-\hat{\x}\|_1 \leq \frac{2C(1+\kappa)}{C-1}\|\x_{\overline{S}}\|_1
\label{eq:robustness}
\eeq
\noindent From \cite{isitrobust} and the scaling law discovered in this paper, we know that for Gaussian
i.i.d. measurement matrices $\A$, $\ell_1$ minimization is weakly robust,
i.e., there exists a robustness factor
$C>1$ as a function of $\frac{|S|}{n} < \mu_{W}(\delta)$ for which
(\ref{def1}) holds. Now let $k_1 = (1-\epsilon_1)\mu_{W}(\delta)n$ for some small
$\epsilon_1>0$, and $K_1$ be the $k_1$-support set of $\x$, namely,
the set of the largest $k_1$ entries of $\x$ in magnitude. Based on
equation (\ref{eq:robustness}) we may write
\beq
\|\x-\hat{\x}\|_1 \leq \frac{2C(1+\kappa)}{C-1}\|\x_{\overline{K_1}}\|_1
\label{eq:robustness1}
\eeq
\noindent For a fixed value of $\delta$, $C$ in (\ref{eq:robustness1})
is a function of $\epsilon_1$ following the scaling law discovered in this paper, and becomes arbitrarily close to $1$ as
$\epsilon_1\rightarrow 0$. $\kappa$ is also a bounded function of
$\epsilon_1$ and therefore we may replace it with an upper bound
$\kappa^*$. We now have a bound on $\|\x-\hat{\x}\|_1$. To explore
this inequality and understand its asymptotic behavior, we
apply a third result, which is a certain concentration bound on
the order statistics of the random variables following certain amplitude distributions.

\begin{lem}
Suppose $X_1,X_2,\cdots,X_N$ are $N$ i.i.d. random variables whose amplitudes, with a mean value of $E(|X|)$, follow the probability density function $f(x)$ for $x\geq 0$. Let $S_N = \sum_{i=1}^{N}|X_i|$ and let $S_M$ be the sum of
the smallest $M$ numbers among the $|X_i|$, for each
$1\leq M \leq N$. Then for every $\epsilon>0$, as $N\rightarrow\infty$,
we have
\bea
&&\Prob(|\frac{S_N}{N} - E(|X|)| > \epsilon)\rightarrow 0, \nonumber \\
&&\Prob(|\frac{S_M}{S_N} -
\frac{1}{E(|X|)}\int_{0}^{F^{-1}(\frac{M}{N})}x f(x) dx|>\epsilon )\rightarrow 0,\nonumber
\eea
\noindent where $F(x)$ is the corresponding cumulative distribution function for the considered random variable amplitude $|X|$.
\label{lemma:Gaussian_Base}
\end{lem}

\noindent Without loss of generality, we assume $E(|X|)=1$. As a direct consequence of Lemma \ref{lemma:Gaussian_Base} we can write:
\beq
\Prob(|\frac{\|\x_{\overline{K_1}}\|_1}{\|\x\|_1} - \int_{0}^{F^{-1}(\frac{\epsilon_0+\epsilon_1}{1+\epsilon_{0}})} xf(x) dx |>\epsilon)\rightarrow0
\label{eq:kbar bound}
\eeq
\noindent for all $\epsilon>0$ as $n\rightarrow\infty$. Define
\beq
\nonumber \zeta(\epsilon_0):= \inf_{\epsilon_1>0} \frac{2C(1+\kappa^*)}{C-1} \int_{0}^{F^{-1}(\frac{\epsilon_0+\epsilon_1}{1+\epsilon_{0}})} xf(x) dx >\epsilon
\eeq

\noindent Combining (\ref{eq:robustness1}) with (\ref{eq:kbar bound}) we can get
\beq
\Prob(\frac{\|\x-\hat{\x}\|_1}{\|\x\|_1} - \zeta(\epsilon_0) < \epsilon)\rightarrow 1
\label{eq:robustness2}
\eeq
\noindent for all $\epsilon>0$ as $n\rightarrow\infty$. In summary, we have showed that
$|K\cap L|\geq k-W(\x,\|\x-\hat{\x}\|_1)$, and then ``weak robustness'' of $\ell_1$ minimization guarantees that for large $n$ with high probability $\|\x-\hat{\x}\|_1 \leq \zeta(\epsilon_0)\|\x\|_1$. These results will further lead to the main claim on the support recovery, which extends a similar claim in \cite{Isit reweighted l_1} by using the closed-form scaling law result in this paper.
\begin{thm}[Support Recovery]
Let $\A$ be an i.i.d. Gaussian $m\times n$ measurement matrix with
$\frac{m}{n}=\delta$. Let $k=(1+\epsilon_0)\mu_{W}(\delta)$ and $\x$
be an $n\times1$ random $k$-sparse vector whose nonzero element amplitude follows the distribution of $f(x)$. Suppose that
$\hat{\x}$ is the approximation to $\x$ given by the $\ell_1$ minimization, namely $\hat{\x}=argmin_{\A\z=\A\x}\|\z\|_1$. Then, for any $\epsilon_{0}>0$ and for all $\epsilon>0$, as $n\rightarrow\infty$,
\beq
\small
\Prob(\frac{|supp(\x) \cap supp_k(\hat{\x})|}{k} -(1-F(y^*))
>-\epsilon)\rightarrow 1,
\label{eq:support recovery}
\eeq
where $y^*$ is the solution to $y$ in the equation $\int_{0}^{y} xf(x) dx=\zeta(\epsilon_0)$.

Moreover, if the integer $t\geq 0$ is the smallest integer for which the amplitude distribution $f(x)$ has a nonzero $t$-th order derive at the origin, namely $f^{(t)}(0) \neq 0$, then as $\epsilon_{0} \rightarrow 0$, with high probability,
\beq
\small
\frac{|supp(\x) \cap supp_k(\hat{\x})|}{k}= 1-O(\epsilon_{0}^{\frac{1}{t+2}}).
\eeq
\label{thm:l_1 support recovery}
\end{thm}
The proof of Theorem \ref{thm:l_1 support recovery} relies on the scaling law for recovery stability in this paper and concentration Lemma \ref{lemma:Gaussian_Base}. Note that if $\epsilon_0\rightarrow0$, then  Theorem \ref{thm:l_1 support recovery} implies that $\frac{|K\cap L|}{k}$ becomes arbitrarily close to 1. We can also see that the support recovery is better when the probability distribution function of $f(x)$ has a lower order of nonzero derivative.  This is consistent with the better recovery performance observed for such distributions in simulations of the iterative reweighted $\ell_{1}$ minimization algorithms.

 \section{Perfect Recovery, Step 3 of the Algorithm}
 \label{sec:perfect recovery}
 In Section \ref{sec:robustness} we showed that. if
 $\epsilon_0$ is small, the $k$-support of $\hat{\x}$, namely
 $L=supp_k(\hat{\x})$, has a significant overlap with the true support of
 $\x$. The scaling law gives a quantitative lower bound on the size of this overlap
 in Theorem \ref{thm:l_1 support recovery}. In Step 3 of
 Algorithm \ref{alg:modmain}, weighted $\ell_1$ minimization is used,
 where the entries in  $\overline{L}$ are assigned a higher weight
 than those in $L$. In \cite{isitweighted}, we have been able to
 analyze the performance of such weighted $\ell_1$ minimization
 algorithms. The idea is that if a sparse vector $\x$ can
 be partitioned into two sets $L$ and $\overline{L}$, where in one set
 the fraction of non-zeros is much larger than in the other set, then
 (\ref{eq:weighted l_1}) can potentially increase the recovery threshold of $\ell_{1}$ minimization.
\begin{thm} \cite{isitweighted}
Let $L\subset \{1,2,\cdots,n\}$ , $\omega>1$ and the fractions
$f_1,f_2\in[0,1]$ be given. Let $\gamma_1 = \frac{|L|}{n}$ and
$\gamma_2=1-\gamma_1$. There exists a threshold
$\delta_c(\gamma_1,\gamma_2,f_1,f_2,\omega)$ such that with high
probability, almost all random sparse vectors $\x$ with \emph{at least}
$f_1\gamma_1n$ nonzero entries over the set $L$, and \emph{at most}
$f_2\gamma_2n$ nonzero entries over the set $\overline{L}$ can be
perfectly recovered using
$\min_{\A\z=\A\x}\|\z_L\|_1+\omega\|\z_{\overline{L}}\|_1$, where $\A$
is a $\delta_cn\times n$ matrix with i.i.d. Gaussian entries. Furthermore, for appropriate $\omega$,
\[
\mu_W(\delta_c(\gamma_1,\gamma_2,f_1,f_2,\omega))<f_1\gamma_1+f_2\gamma_2,
\]
i.e., standard $\ell_1$ minimization using a $\delta_cn\times n$
measurement matrix with i.i.d. Gaussian entries cannot recover such $x$.
\label{thm:delta}
\end{thm}
A software package for computing such thresholds can also be found in
\cite{sotware link}. We then summarize the threshold improvement result in the following theorem, with the detailed proofs omitted due to limited space.

\begin{thm}[Perfect Recovery]
Let $\A$ be an $m\times n$ i.i.d. Gaussian matrix with
$\frac{m}{n}=\delta$. If $\delta_c(\mu_{W}(\delta),1-\mu_{W}(\delta),1,0,\omega) < \delta$,
then there exist $\epsilon_0>0$ and $\omega>0$ such that, with high probability as $n$ grows to infinity, Algorithm
\ref{alg:modmain} perfectly recovers a random
$(1+\epsilon_0)\mu_{W}(\delta)n$-sparse vector with i.i.d. nonzero
entries following an amplitude distribution whose pdf has a nonzero derive of some finite order at the origin.
\label{thm: final thm}
\end{thm}
\section{Acknowledgement}
 We thank A. Khajehnejad, Babak Hassibi and A. Avestimehr for the helpful discussions. The research is supported by NSF under CCF-0835706.

\end{document}